\documentclass[12pt,reqno,paper=letter,abstracton]{scrartcl}
\usepackage[authoryear]{natbib}
\usepackage[utf8]{inputenc}
\usepackage{lmodern}

\usepackage{ccfonts}
\usepackage[euler-digits,T1]{eulervm}
\usepackage[T1]{fontenc}
\usepackage{textcomp}
\usepackage{microtype}
\usepackage{amsmath,amsthm,amssymb}
\allowdisplaybreaks[4]
\usepackage{mathtools}
\usepackage{url}
\usepackage{tikz}
\usepackage{float}
\usepackage{verbatim}
\usepackage[obeyDraft,linecolor=orange,backgroundcolor=yellow]{todonotes}
\usepackage[colorlinks=false, pdfborder={0 0 0}]{hyperref}
\renewcommand{\cite}{\citep}

\newcommand{\mC}{\mathcal{C}}
\newcommand{\uz}{\underline{0}}
\newcommand{\ud}{d}
\newcommand{\utheta}{\theta}
\newcommand{\etas}{\eta \,}

\newcommand{\mL}{\mathcal{L}}
\newcommand{\mN}{\mathcal{N}}

\newcommand{\mG}{\mathcal{G}}
\newcommand{\mK}{\mathcal{K}}

\newcommand{\mF}{\mathcal{F}}

\theoremstyle{plain}
\newtheorem{theorem}{Theorem}[section]
\newtheorem{lemma}[theorem]{Lemma}

\theoremstyle{definition}

\theoremstyle{remark}

\begin{document}
%\begin{spacing}{1.1}
%\setlength{\jot}{1.5ex}
%\maketitle
%\markboth{A. Sinha and A. Anastasopoulos}{General Mechanism Design Methodology for Social Utility Maximization with Linear constraints}

\title{A General Mechanism Design Methodology for Social Utility Maximization with Linear Constraints}
%\subtitle{Fictitious Play analysis and the relation between FTPL and FTRL}
%\author{Abhinav Sinha
%\affil{University of Michigan}
%Achilleas Anastasopoulos
%\affil{University of Michigan} }
%\title{Generalized Proportional Allocation Mechanism Design for Multi-rate Multicast  Service on the Internet}
\author{Abhinav Sinha and Achilleas Anastasopoulos
\\ {\normalsize EECS Department, University of Michigan, Ann Arbor, MI 48109}
\\[-0.5ex] \normalsize \texttt{\{absi,anastas\}@umich.edu}} %\\
\date{\normalsize \today}

\maketitle

\begin{abstract}
Social utility maximization refers to the process of allocating resources in such a way that the sum of agents' utilities is maximized under the system constraints.
Such allocation arises in several problems in the general area of communications, including unicast (and multicast multi-rate) service on the Internet, as well as in applications with (local) public goods, such as power allocation in wireless networks, spectrum allocation, etc.
Mechanisms that implement such allocations in Nash equilibrium have also been studied but either they do not possess full implementation property, or are given in a case-by-case fashion, thus obscuring fundamental understanding of these problems.

In this paper we propose a unified methodology for creating mechanisms that fully implement, in Nash equilibria, social utility maximizing functions arising in various contexts where the constraints are convex.
The construction of the mechanism is done in a systematic way by considering the dual optimization problem.
In addition to the required properties of efficiency and individual rationality that such mechanisms ought to satisfy,
three additional design goals are the focus of this paper: a) the size of the message space scaling linearly with the number of agents (even if agents' types are entire valuation functions), b) allocation being feasible on and off equilibrium, and c) strong budget balance at equilibrium and also off equilibrium whenever demand is feasible.\let\thefootnote\relax\footnote{\small \textit{Terms:} Mechanism Design, Information Elicitation, Foundations of Incentive Compatibility, Auction Theory.}\let\thefootnote\relax\footnote{\small \textit{Key Words and Phrases:} Multicast routing, proportional allocation, game theory, mechanism design, full implementation, Nash equilibrium.}
\end{abstract}

%\category{TFAI}{Theoretical Foundations and Artificial Intelligence and Applied Game Theory}{Theoretical Foundations and Artificial Intelligence and Applied Game Theory}

%\terms{Mechanism Design, Information Elicitation, Foundations of Incentive Compatibility, Auction Theory}

%\keywords{Mechanism design, Full implementation, Game theory, Nash equilibrium, Proportional allocation, Social utility maximization, convex constraints, Strong budget balance, Unicast service, Multi-rate multicast, public goods, power  allocation}

%\acmformat{Gang Zhou, Yafeng Wu, Ting Yan, Tian He, Chengdu Huang, John A. Stankovic,
%and Tarek F. Abdelzaher, 2010. A multifrequency MAC specially
%designed for  wireless sensor network applications.}

\section{Introduction}

%\todo[inline]{Take out Microeconomics from below and abstract}

In the general area of communications, a number of decentralized resource allocation problems have been studied in the context of mechanism design. Such problems include unicast service on the Internet \cite{hajek}, \cite{basar}, \cite{demoscorrection},  multi-rate multicast service on the Internet \cite{demosmulticorrection}, \cite{SiAn_multicast_arxiv}, power allocation in wireless networks \cite{demosshruti}, \cite{han2012game}, \cite{zhang2011game}, spectrum allocation \cite{randall} and pricing in a peer-to-peer network \cite{neely2009optimal}.
The mechanism design framework is both appropriate and useful in the above problems since these problems are motivated by
the designer's desire to allocate resources efficiently in the presence of strategic agents who possess private information about their level of satisfaction from the allocation.

Usually mechanism design solutions define a contract such that the induced game
has at least one equilibrium (Nash, Bayesian Nash, dominant strategy, etc) that corresponds to the desired allocation.
This is usually obtained with \emph{direct} mechanisms by appealing to the revelation principle \cite{tilmanbook} \cite{narahari1}.
The drawbacks of these direct approaches is that they require agents to quote their types
(which may be entire valuation functions) and that the induced game may have other extraneous equilibria that
are not efficient.
In this paper, the focus is on  \textbf{full Nash
implementation}. Without going into a formal definition, full Nash implementation refers
to the design of contracts such that only the designer's most preferred outcome is realized
as a result of interaction between strategic agents (i.e. at Nash equilibria (NE)), as opposed to
general mechanism design, where other less preferred outcomes are also possible.
Thus implementation is more stringent and is capable of producing better allocations. Readers may refer to \cite{jackson} for a survey on implementation theory.

Concentrating on those proposed solutions in the literature that guarantee full Nash implementation, one observes a fragmented and case-by-case approach. One may ask how fundamentally different are for example the problems of unicast service on the Internet,
multicast multi-rate service on the Internet, power allocation in wireless networks, etc,
to justify a separately designed mechanism for each case.
Alternatively, one may ask what are the common features in all these problems that can lead to a unified
mechanism design approach.
These questions provide the motivation for the work presented in this paper.

In particular, our starting point is to state a class of problems as social utility maximization
under linear  inequality/equality constraints. We then proceed by characterizing their solution through Karush-Kuhn-Tucker (KKT) conditions.
Analyzing the dual optimization problem is essential in our approach, because it hints at how
the taxation part of the mechanism should be designed.
Subsequently we define a general mechanism and show that it results in full Nash implementation
when the agents' valuation functions are their private information\footnote{Use of Nash equilibrium as a solution concept itself requires some justification, since it applies only to games with complete information. This issue is discussed further in Section~\ref{secoffeq}. However, in this work we accept the justification given by Reichelstein and Reiter in~\cite{reiter} and Groves and Ledyard in~\cite{groves} based on the second interpretation of Nash equilibrium as the fixed point of an (possibly myopic) adjustment process, \cite{nashthesis51}.}.
Since the application domain of interest is in the area of communications, we give special emphasis on the size of the message space (as a consequence, VCG-type mechanisms, see \cite[Section~5.3]{krishna},  \cite[Chapter~5]{tilmanbook}, \cite[Chapter~10]{shoham2009multiagent}, \cite{narahari1}, \cite{narahari2}, are inappropriate since they require quoting of types, which are entire valuation function in this set-up). A notable exception is the work~\cite[Section~5]{johari2009efficiency} and \cite{hajekvcg} that adapts VCG for a small message space (one dimensional message per user) and guarantees off-equilibrium feasibility. However in their work, full implementation can't be guaranteed and there is possibility of extraneous equilibria.
In this paper, the message space scales linearly with number of agents so that the proposed mechanism is scalable.

In addition to the stated goal of getting optimal allocations at all Nash equilibria, there
are other auxiliary properties that are sought in a mechanism. An important
such property is  \emph{individual rationality}; this requires that agents are weakly better
off at NE than not participating in the mechanism at all. The purpose of this is to ensure
that agents are willing to sign the contract in the first place.
Another important property is \emph{strong budget balance} (SBB) at NE. This requires that at NE, the total monetary transfer between agents (assuming quasi-linear utility functions) is zero.
The unified mechanism proposed in this work possesses both these properties.

Finally, a mechanism may be endowed with auxiliary properties off equilibrium. These are meant
to improve the applicability in practical settings.
For instance, the NE is interpreted as the convergent outcome when agents in the system ``learn'' the
game by playing it repeatedly, which implies that during this process the
messages (and thus allocations and taxes) are off equilibrium.
In our opinion, the most important auxiliary property is \textbf{feasibility of allocation on and off equilibrium}.
This property is essential whenever the system constraints are hard constraints on resources and cannot be violated by any means. For instance, the contract should not promise rate allocations to agents that ever violate the capacity constraints of the network links because such a contract would be practically invalid since it promises something that can never be achieved.
Note that, at equilibrium the allocation has to satisfy the constraints by definition, so feasibility always holds at NE.
Another auxiliary property is SBB off equilibrium. Similarly this property guarantees that in a practical situation where a dynamic learning process converges to NE, each step of this process leaves a zero balance.
One of the main features of the work in this paper is that the allocation scheme is designed to guarantee both
feasibility off equilibrium and SBB off equilibrium whenever demand is feasible.
Specifically, feasibility is achieved by utilizing \emph{proportional allocation}.
Proportional allocation refers to the idea of using agents’ demands to create their allocation by projecting their overall
demand vector back down to the boundary of the feasible region whenever it is outside of it.
This way everyone receives allocation that is proportional to their original
demand. This allocation method generalizes the idea of proportional
allocation that was introduced in mechanism design framework by \cite{basar}, \cite{hajek} for unicast network with one capacity constraint and for stochastic control of networks in~\cite{kelly}, as well as by the authors in~\cite{SiAn_unicast_arxiv}, \cite{SiAn_multicast_arxiv}.

The main contributions of this work are summarized as follows.
\begin{itemize}
\item A unified framework is proposed for full implementation (in NE) mechanism design in social utility maximization problems with convex constraints.

\item The proposed mechanism has small message  space (although agents' types are entire valuation functions) and is scalable.

\item The proposed mechanism is individually rational and strongly budget balanced at NE.

\item The proposed mechanism lends itself to practical implementation (through learning algorithms) because allocation is feasible even off equilibrium (and is strongly budget balanced off equilibrium when demand is feasible).

\end{itemize}

The rest of the paper is structured as follows.
In Section~\ref{secCP} we describe three examples relevant in communications and formulate
the general centralized allocation problem that we wish to implement.
In Section~\ref{secmech} we describe the proposed mechanism.
Section~\ref{secres} contains all the proofs of full implementation. Section~\ref{secoffeq} contains off-equilibrium results and discusses their relevance.
Finally in Section~\ref{secgen} we conclude with a discussion of important generalizations of this set-up.
%Due to space limitations, most of the proofs are relegated to the Appendix at the end of this document.
%An extended version of this paper can be found as a technical report in \cite{ShAn_ec14}.

%\cite{MWG}

%\cite{williams2012communication}

\section{Motivation and Centralized Problem} \label{secCP}

In this section we start by describing various important resource allocation problems that arise in communications and to which our generalized methodology applies (in Section~\ref{secgen} we discuss generalizations that will help solve an even larger class of problems than the ones described below). After the examples, we define a general centralized optimization problem, \eqref{CPg}, that covers all the examples. Following that we state general assumptions on~\eqref{CPg}
(some additional technical assumptions will be made in Section~\ref{secres}).

\subsection{Interesting Resource allocation problems in Communications} % and Microeconomics}
%\todo{We can eliminate microeconomics from the title}

%Below we describe three important resource allocation problems arizing in communications and microeconomics, which are of the form of~\eqref{CPg}.

\subsubsection*{Unicast Transmission on the Internet}

Consider agents on the Internet from set $ \mN = \{1,\ldots,N\} $, where each agent $ i \in \mN $ is a pair of source and destination that communicate via a pre-decided route consisting of links in $ \mL_i $. All the agents together communicate on the network consisting of links in $ \mL = \cup_{i \in \mN} \mL_i $. Since each link in the network has a limited capacity, this results in constraints on the information rate allocated to each agent. Considering a scenario where agents have (concave) utility functions/profiles $ \{v_i(\cdot)\}_{i \in \mN} $ that measure the satisfaction received by agents for various allocated rates, we can write the social utility optimization problem as
\begin{gather} \label{CPuni}
\max_{x \in \mathbb{R}_+^N}~ \sum_{i \in \mN} v_i(x_i) \tag{CP$ _\text{u} $}\\
\text{s.t.} \quad \sum_{i \in \mN^l} \alpha_i^l x_i \le c^l \quad \forall~ l \in \mL. %\tag{C\textsubscript{u,1}}
\end{gather}
In the above $ \mN^l \triangleq \{ i \in \mN \mid l \in \mL_i \} $ is the set of agents present on link $ l $, $ c^l $ is the capacity of link $ l $ and $ \alpha_i^l $ are positive weights meant to differentiate between true information rate $ x_i $ and its imposition on capacity of the links of the network through coding rate, packet error etc.

The above feasible set is a polytope in the first quadrant of $ \mathbb{R}_+^N $ and is created by faces that have outward normal vectors pointing away from the origin. For the details of a full implementation mechanism specifically for the unicast problem readers may refer to \cite{rahuljain},  \cite{demoscorrection}, \cite{SiAn_unicast_arxiv}.

%\cite{rahuljain}, \cite{basar}, \cite{hajek}
%\todo{Comment on proportional idea}
%\todo{not sure if technical description is needed rightaway}

%The above problem~\eqref{CPuni} can be cast as~\eqref{CPg} with $ A_{li} = \alpha_i^l $ whenever $ i \in \mN^l $ and $ A_{li} = 0 $ otherwise. Clearly there is no need for auxiliary variables, so $ B_{li} = 0 $.

%\begin{comment}

%\end{comment}

\subsubsection*{Public Goods}

In contrast to the private consumption problem above, there are public goods problems where the resources are shared directly between agents instead of sharing via constraints. The unicast problem was a scenario where allocation of rate to one agent on a link would imply less available rate for other agents on the link (because of the capacity constraint) - the rate allocation thus is private in this case. In contrast to this, if there is a resource allocation problem where two or more agents can simultaneously use the same resource then it will be classified as a public goods problem. A well-studied example of this kind is the wireless transmission with interference, in which the  power level of each agent affects any other agent through the  signal-to-interference-and-noise ratio (SINR).

Here we consider a simplified scalar version of the  public goods problem for a system of agents from the set $ \mN $ in the following form
\begin{gather}
\max_{x \in \mathbb{R}_+}~ \sum_{i \in \mN} v_i(x) \tag{CP$ _\text{pb} $}\\
\text{s.t.} \quad x \in \mathcal{X},
\end{gather}
where $ \mathcal{X} $ is a convex subset of $ \mathbb{R}_+ $. Note that the argument for all utility functions is the same; since there is one public good being simultaneously used by all the agents (for which $ x $ marks the usage level). We rewrite such a problem in the constrained form%It is rather straightforward to formulate this in the form of~\eqref{CPg},
%\begin{gather}
%\max_{\underline{x} \in \mathbb{R}_+^N}~\sum_{i \in \mN} v_i(x_i) \tag{CP$ _\text{pb} $} \\
%\text{s.t.} \quad x_1 = x_2 = \cdots = x_N \\
%\text{s.t.} \quad x_1 \in \mathcal{X}
%\end{gather}
%Furthermore, each linear equality constraint can easily be written as two linear inequality constraints. In a compact form, this gives us
\begin{subequations}
\begin{gather}
\max_{\underline{x} \in \mathbb{R}_+^N}~\sum_{i \in \mN} v_i(x_i) \tag{CP$ _\text{pb} $} \\
\text{s.t.} \quad x_1 = x_2 = x_3 = \ldots = x_N \\
\text{s.t.} \quad x_1 \in \mathcal{X}.
\end{gather}
\end{subequations}
The treatment when $ x $ is a vector is not very different and we discuss this in the generalizations at the end. Implementation for the public goods problem is the most studied of all the examples in this paper, see \cite{groves}, \cite{hurwicz1979outcome}, \cite{chen2002family}. The distinction from private goods is that generally public goods problems require handling the ``free-rider'' problem \cite[Section 11.C]{MWG}.  %and \cite{demosshruti}.
%\todo{more references here}

\subsubsection*{Local Public Goods}

Another important resource allocation problem in communications is the local public goods problem. The basic idea of direct sharing of resources is same as above but in here the sharing is only among agents locally. So there are local groups of agents for whom the allocation has to be the same, but this common allocation can be different from one group to the next. If we divide the set of agents into disjoint local groups: $ \mN = \sqcup_{k \in \mK} \mN_k  $, then the centralized problem can be written as
\begin{gather} \label{CPlpg}
\max_{{x} \in \mathbb{R}_+^K}~ \sum_{k \in \mK} \sum_{i \in \mN_k} v_i(x_k) \tag{CP$ _\text{lpb} $}\\
\text{s.t.} \quad x_k \in \mathcal{X}_k ~~\forall~k \in \mK,
\end{gather}
with $ (\mathcal{X}_k)_{k \in \mK} $ being convex subsets of $ \mathbb{R}_+ $. As before, we would restate it as
\begin{subequations}
\begin{gather}
\max_{{x} \in \mathbb{R}_+^K}~ \sum_{k \in \mK} \sum_{i \in \mN_k} v_i(x_i) \tag{CP$ _\text{lpb} $}\\
\text{s.t.} \quad x_i = x_j ~~\forall~i,j \in \mN_k,~\forall~k \in \mK \\
\text{s.t.} \quad x_{j_k} \in \mathcal{X}_k ~\text{for some}~j_k \in \mN_k,~\forall~k \in \mK.
\end{gather}
\end{subequations}
Local public goods problems are relevant in those network settings where there is direct interaction between local agents. Wireless transmission is an example of this, each agent affects and is affected by other agents' transmission through interference (SINR) and it is reasonable to assume that this effect is only local and that agents situated far enough (either spatially or in the frequency domain) will not affect each other. Readers may refer to \cite{demosshruti} for a specific mechanism for the local public goods problem.

\subsection{General Centralized Problem} \label{subsecCP}

Here we state the generic form of the centralized optimization problem that we wish to fully implement.
The resource allocation problem is defined for a system with agents indexed in the set $ \mN = \{1,2,\ldots,N\} $, who have utility functions $ \{v_i(\cdot)\}_{i \in \mN} $. The objective is to find the optimum allocation of a single infinitely divisible good that maximizes the sum of utilities %(Pareto efficient)
subject to constraints on the system. The allocation made to agents will be denoted by the vector $ x \in \mathbb{R}_+^N $, with $ x_i $ being the allocation to agent $ i $. %In addition, to make the problem general enough, we have auxiliary variables $ m \in \mathbb{R}^M $. They do not directly form a part of any agents' allocation but are used purely to state the constraints on allocation.
The centralized optimization problem that we consider is
\begin{gather} \label{CPg}
\max_{x} \sum_{i\in \mN} v_i(x_i) \tag{CP} \\
\text{s.t.} \quad x \in \mathbb{R}_+^N \tag{C$ _\text{1} $}\\
\text{s.t.} \quad A^\top_l x \le c_l ~~\forall~l \in \mL ~~\text{where}~~ A_l \in \mathbb{R}^N, c_l \ge 0. \tag{C$ _\text{2} $}
\end{gather}
%\todo[inline]{$ A_l \ge 0 $ for good measure, can be changed if needed}
%\todo[inline]{Multicast is missing, I am not sure yet, it might convolute things too much}
The set $ \mL = \{1,2,\ldots,L\} $ indexes all the constraints and $ A_l,c_l $ are all parameters of the optimization problem. It is easy to see that the above set-up covers equality constraints (such as from the public goods example), since we can always write $ x_1 = \ldots = x_N $ as $ x_1 \ge x_2 \ge \ldots \ge x_N \ge x_1 $. %Note that as in the Multicast problem, the auxiliary variables are used here only to state the feasibility conditions. So the purpose of $ m $ is to check the feasibility of candidate optimal solutions $ x^\star $.

Denote by $ \mC \subset \mathbb{R}_+^{N} $ the above feasible set. Note that $ \mC $ is a polytope in the first quadrant of $ \mathbb{R}^{N} $, possibly of a lower dimension than $ N $ (due to equality constraints). For convenience, we denote by $ \mL_i $ the set of constraints that involve agent $ i $, i.e. $ \mL_i = \{l \in \mL \mid A_{li} \ne 0 \} $ with $ L_i  \triangleq \vert \mL_i \vert $.  Conversely, define $ \mN^l $ as the set of agents involved at link $ l $ i.e. $ \mN^l = \{i \in \mN \mid l \in \mL_i \} $ with $ N^l \triangleq \vert \mN^l \vert $.

%\subsection{Interpretation of Constraints}

%We can divide the constraints in~(C$ _\text{2} $) into three distinct categories - A, B and C, so $ \mL = \mL_A \cup \mL_B \cup \mL_C $. Category A is of those constraints that contain more than one $ x_i $'s i.e. $ A_{li},A_{lj} \neq 0 $ for some $ i,j \in \mN $. Category B is of those constraints that involve exactly one $ x_i $ and category C is of those constraints that involve none of the $ x_i $'s i.e. $ A_{li} = 0 $ for all $ i \in \mN $.
%Let $ (x^\star,m^\star) $ be the solution of~\eqref{CPg}. Also denote by $ \mN_j $ the set of agents that share a constraint from~(C$ _\text{2} $) with auxiliary variable $ m_j $ i.e. $ \mN_j = \{ i \in \mN \mid \exists \: l \in \mL \text{ s.t. } A_{li} \neq 0, B_{lj} \neq 0 \} $. \ldots

%The vectors $ \left( A_l,B_l \right)_{l \in \mL} $ define the feasibility set $ \mC $. Let $ \mL_i \triangleq \{l \in \mL \,\mid\, A_{li} \neq 0 \} $ denote the set of constraints that involve agent $ i $. In order to have agents quote prices for various constraints we need to refine this further. \ldots %For reasons that be clear later, we define an extension of the set $ \mL $

\subsection{Assumptions} \label{subsecassump}

Stated below are the assumptions on~\eqref{CPg} some of which restrict the environment $ \{v_i\}_{i\in \mN} $
and some the constraint set $ \mC $. Some additional technical assumptions will be introduced later to handle
the degenerate cases.

\begin{itemize}
\item[(A1)] \noindent\fbox{\parbox{0.9\textwidth}{For any $ i \in \mN $, $ v_i(\cdot) $ %\in \mathcal{V}_i \subseteq \mathcal{V}_0 $, the set of
is a strictly concave and continuously double differentiable function $ \mathbb{R}_+ \rightarrow \mathbb{R} $.}}
\end{itemize}
The purpose of strict concavity is to have a~\eqref{CPg} whose solution can be described sufficiently by the KKT conditions (note that monotonicity is not assumed).

\begin{itemize}
\item[(A2)] \noindent\fbox{\parbox{0.9\textwidth}{The optimal solution $ x^{\star} $ is bounded such that $ x^{\star} \in \times_{i=1}^{N} \: (d_i,D) $ for some $ 0 < d_i < D $, with $ \ud = \left(d_i\right)_{i=1}^N \in %\text{int}
\mC $ being arbitrarily close to $ \uz $ and $ D $ being large enough.}} % \hfill (A4)
\end{itemize}
This assumption is used to eliminate corner cases of~\eqref{CPg}, since they usually require special treatment and make the exposition unnecessarily convoluted. Note that we can always select a point $ \ud \in \mC $ that is arbitrarily close to $ \underline{0} $ because of assumption (A3) below.
%With this assumption, one can equivalently solve~\eqref{CPg} with the constraint set $ \mC = \mC \cap [0,2D]^{N} $. The sole purpose for this modification is to make the constraint set compact. Without this we could potentially have an unbounded constraint set; but since we are interested only in resource allocation problems where the solutions are absolutely bounded, modifying the constraint set to make it bounded will not change the solution. Note that for cases where $ \mC $ is already compact, this modification will not alter the constraint set. The new constraints are also linear and hence of the same form as the ones defining $ \mC $. Also, due to the assumption, the optimal solution can never be such that the new constraints of the form $ x_i \le 2D $ are active.

The next two assumptions restrict the constraint set $ \mC $.
\begin{itemize}
\item[(A3)] \noindent\fbox{\parbox{0.9\textwidth}{The vector $ \underline{0} \in \mC $, i.e. $ x = \underline{0} $ is feasible.}} %\hfill (A3)
\end{itemize}
Since we are considering problems where all the variables $ x_i $ have a physical interpretation, it is natural to consider a constraint set whereby every agent getting $ 0 $ allocation is feasible. Note that this also explains the choice $ c_l \ge 0 $.

%\item[(A2)] For any $ i \in \mN $, $ v_i^\prime(0) $ is finite. Since $ v_i(\cdot) $ are strictly concave, this also means that $ v_i^\prime(x) $ is bounded for any $ x \in \mathbb{R}_+ $.

\begin{itemize}
\item[(A4)] \noindent\fbox{\parbox{0.9\textwidth}{For any constraint $ l \in \mL $ in (C$ _\text{2} $) there are two distinct $ i,j \in \mN $ such that $ A_{li},A_{lj} \ne 0 $. This can also be stated as $ N^l \ge 2 $ $ \forall $ $ l \in \mL $ where $ N^l $ is as defined after~\eqref{CPg}.}}
%\begin{equation}
%\mN^l = \{ i \in \mN \mid A_{li} \ne 0 \}  \equiv \{ i \in \mN \mid l \in \mL_i \} \quad \text{and} \quad   N^l \triangleq \vert \mN^l \vert
%\end{equation}
\end{itemize}
This ensures that there is indeed competition for all the constraints that could possibly be active at optimum. Again this is used to avoid special treatment of corner cases.
%\item[(A6)] Each auxiliary variable is associated with at least two agents in the constraints from (C\textsubscript{2}) i.e. $ \vert \mN_j \vert \ge 2 $ for all $ j \in \mM $.
%\item[(A3)] No agent receives zero allocation at optimum i.e. $ x_i^{\star} > 0 $ $ \forall~i \in \mN $.

Denote by $ \mathcal{V}_0 $ the set of all possible functions $ \{v_i(\cdot)\}_{i \in \mN} $ that satisfy the above assumptions. Then $ \mathcal{V}_0 $ will be the environment for our mechanism design problem.

We will also make an assumption on the overall utility of agents in the system.
\begin{itemize}
\item[(A5)] \noindent\fbox{\parbox{0.9\textwidth}{Apart from the valuation part $ v_i(x_i) $ there is linear taxation component as well that affects agents' utilities. So overall utility of agent $ i $ is
\begin{equation}
u_i(x,t) = v_i(x_i) - t_i. %\tag{A6}
\end{equation}}}
\end{itemize}

\subsection{KKT conditions} \label{subseckkt}
The Lagrangian for the optimization problem~\eqref{CPg} is
\begin{equation}
L(x,\lambda,\mu) = \sum_{i \in \mN} v_i(x_i) - \sum_{l \in \mL} \lambda_l \left( A^\top_l x - c_l \right) + \sum_{i \in \mN} \mu_i x_i
\end{equation}
Due to assumption (A2), we can state the KKT conditions only in terms of $ \lambda^{\star} $ and not involve $ \mu^{\star} $. %Also we don't need to include terms above relating to the modified constraint set $ \mC $, since in (A4) we also assumed $ x_i < D $.
\begin{enumerate}
\item \textsc{Primal Feasibility}: $ x^{\star} \in \mC $.

\item \textsc{Dual Feasibility}: For all $ l \in \mL $, %and $ j \in \mM $,
$ \lambda_l^{\star} \ge 0 $.
%\nu_j^{\star} \ge 0
%\qquad \mu_i^{\star} \ge 0 ~~\forall~i \in \mN \qquad \nu_j \ge 0  ~~\forall~j \in [M]

\item \textsc{Complimentary Slackness}: For all $ l \in \mL $,
\begin{gather}
\lambda_l^\star \left( A^\top_l x^\star - c_l \right) = 0. %\quad \forall~l \in \mL
\end{gather}
%For all $ j \in \mM $
%\begin{equation}
%\nu_j^\star m_j^\star = 0
%\end{equation}
\item \textsc{Stationarity}: For all $ i\in \mN $,
\begin{align}
v_i^\prime(x_i^\star) = \sum_{l \in \mL} A_{li} \lambda_l^\star.  %\quad \text{if}~~ x_i^\star > 0 \\
%\le \sum_{l \in \mL} \lambda_l^\star A_{li} \quad \text{if}~~ x_i^\star = 0
\end{align}
%and for all $ j \in \mM $,
%\begin{gather}
%\sum_{l \in \mL} B_{lj} \lambda_l^\star  = 0. %~~\forall~j \in \mM
%\end{gather}
\end{enumerate}

We  will see later that taxation will help us in achieving these KKT conditions when agents play the induced game from the mechanism. For that it will be important to think of the Lagrange multipliers above as ``prices'' where there will be one price per constraint.

\section{Mechanism} \label{secmech}

In this section we describe the proposed mechanism. Description of the mechanism is divided into two parts - allocation and taxes.
A mechanism in the Hurwicz-Reiter framework consists of an environment, an outcome space, a (centralized) correspondence between the environment and the outcome space, a message space and a contract from the message space to the outcome space. In our case the environment is set $ \mathcal{V}_0 $ of all possible valuation functions $ \{v_i(\cdot)\}_{i \in \mN} $. The outcome space is the Cartesian product of the set of all possible allocations and taxes, which is the set $ \mC \times \mathbb{R}^N \subset \mathbb{R}_+^N \times \mathbb{R}^N $. The correspondence between $ \mathcal{V}_0 $ and $ \mC $ is provided implicitly by the centralized  problem~\eqref{CPg}, where for each $ \{v_i(\cdot)\}_{i \in \mN} $ we get an optimal allocation $ x^\star $ by solving~\eqref{CPg} (explicitly one can solve KKT to define $ x^\star $, together with corresponding Lagrange multipliers). This leaves the designer with the task of designing the message space and the contract.

%Define $ G = [-D,-d] \cup [d,D] $ for some positive parameters $ d,D $ (defined later).

The \textit{message space} for our mechanism is $ \mathcal{S} = \times_{i \in \mN} \mathcal{S}_i $ with $ \mathcal{S}_i = (d_i,+\infty) \times \mathbb{R}_+^{L_i} $ where messages from agents are of the form $ s_i = (y_i, p_i) $ with $ p_i = (p_i^l)_{l \in \mL_i} $ and the total message is denoted by $ s = (s_i)_{i\in\mN} = (y,P) $ with $ y = (y_i)_{i \in \mN} $ and $ P = (p_i)_{i \in \mN} $. The message $ s_i = (y_i,p_i) $ is to be interpreted as follows: $ y_i $ is the level of demand from agent $ i $ and $ p_i $ is the vector of prices that he believes everyone else should pay for the respective constraints. The \textit{contract} $ h: \mathcal{S} \rightarrow \mathbb{R}_+^N \times \mathbb{R}^N $ will specify allocation and taxes for all agents based on the message $ s $, i.e.,  $ h(s) = \left(h_{x,i}(s), h_{t,i}(s)\right)_{i \in \mN} $.

This contract along with agents' utilities will give rise to a one-shot game $ \mathfrak{G} = \left(\mN, \times_{i\in \mN}\mathcal{S}_i, \{\widehat{u}_i\}_{i \in \mN}\right) $ between agents in $ \mN $ where action sets are $ \mathcal{S}_i $ and utilities are
\begin{equation}
\widehat{u}_i(s) = u_i(x,t) = v_i(x_i) - t_i = v_i(h_{x,i}(s)) - h_{t,i}(s).
\end{equation}

\paragraph{Information assumptions} We assume that for any agent $ i $, $ v_i(\cdot) $ is his private information. The mechanism designer doesn't know $ \{v_i(\cdot)\}_{i \in \mN} $ but knows the set $ \mathcal{V}_0 $ to which they belong. Also the constraints (C$ _\text{1} $) and (C$ _\text{2} $) in~\eqref{CPg} (along with the assumptions) are common knowledge i.e. known to agents and the mechanism designer.

We say that the mechanism fully implements the centralized problem~\eqref{CPg} if \textbf{all} Nash equilibria of the induced game $ \mathfrak{G} $ correspond to the unique allocation $ x^\star $ - solution of~\eqref{CPg}, and additionally individual rationality is satisfied i.e. agents are weakly better-off participating in the contract at equilibrium than not participating at all. %All this will be despite the fact that the functions $ h_{x,i} $ and $ h_{t,i} $ will be defined without the knowledge of $ \{v_i(\cdot)\}_{\mN} $, since the designer doesn't have that information.

\subsection{Allocation} \label{secalo}

We first describe the allocation in the case where the constraints in (C$ _{2} $) do not have any effective degeneracy i.e. equality constraints. This distinction is based on whether the feasible set $ \mC $ has a proper interior or not. Clearly when there are no equality constraints, the constraint set will have an interior.

For allocation in this case, we first choose a point $ \utheta $ in the interior of the feasible set such that
\begin{equation}
\utheta \in \mathrm{int}(\mC) \cap  \times_{i=1}^{N} \: (0,d_i)
\end{equation} %Here $ \theta $ (and consequently $ \epsilon $) are chosen so that $ 0 < \theta_i < d $.
%where $ \mathrm{int}(\mC) $ is to be interpreted only as a partial interior of the feasibility set $ \mC $. For example, in case of equality constraints the constraint set is degenerate in a few dimensions, so $ \utheta $ is allowed to be on the boundary in those dimensions, but has to be strictly in the interior for the remaining dimensions.
Note that we can guarantee the existence of $ \utheta $ since by assumption (A3), $ \underline{0} \in \mC $ and clearly $ \mC $ being intersection of half-planes, is a convex set. Since $ \ud $ can be made arbitrarily close to $ \uz $, the same holds for $  \utheta $ as well. % (we will elaborate on the relation between $ \ud $ and $ \utheta $ in the proof of Lemma~\ref{lemst}). \todo[inline]{Probably not needed}
%In lieu of the alternate reformulation of constraints presented in Section~\ref{subsecCP} when there are equality constraints in~\eqref{CPg}, for convenience we denote by $ \widetilde{\utheta} $, those components from $ \utheta $ that correspond with reduced feasible set $ \widetilde{\mC} $ in the $ \widetilde{x}- $ space.
%We start by choosing a point in the interior of the feasible set, say $ \theta = (\theta_x, \theta_m) \in \mC \backslash \partial \mC $.

Before formally defining the allocation, we define it informally with the help of Fig.~\ref{figone}.
\begin{figure}
\centering
%\begin{center}
\begin{tikzpicture}
\draw[line width=1.5pt] [->] (-0.5,0) -- (7,0) ;
\draw[line width=1.5pt] [->] (0,-0.5) -- (0,7) ;
\draw[line width=1pt] (0,0) -- (1.5,4.5) ;
\draw[line width=1pt] (1.5,4.5) -- (5,5) ;
\draw[line width=1pt] (5,5) -- (4.5,1) ;
\draw[line width=1pt] (4.5,1) -- (0,0) ;
\draw[line width=1pt,dotted] (1,1) -- (1,6) ;
\node[above] at (1,6) {\small$ x_1 = d_1 $} ;
\draw[line width=1pt,dotted] (1,1) -- (6,1) ;
\node[right] at (6,1) {\small$ x_2 = d_2 $} ;
\draw[fill=black] (0.5,0.5) circle (0.04);
\node at (0.3,0.3) {$ \mathbold{\theta} $};
%\draw[line width=0.2pt, dashed] (0.5,0.5) -- (2.6,5.4) ;
%\draw[line width=0.2pt, dashed] (0.5,0.5) -- (6,6) ;
%\draw[line width=0.2pt, dashed] (0.5,0.5) -- (5.4,1.9) ;
%\draw[line width=0.2pt, dashed] (0.5,0.5) -- (-0.5,-0.5) ;
\draw (0.5,0.5) -- (5.5,2.5) ;
\draw (0.5,0.5) -- (1.5,3) ;
\draw[fill=black] (5.5,2.5) circle (0.04);
\node[above] at (5.5,2.5) {\small$ y $} ;
\draw[fill=black] (4.6447,2.1578) circle (0.04);
\node[below] at (4.6447,2.1578) {\small$ ~~~x $} ;
\draw[fill=black] (1.5,3) circle (0.04);
\node[above] at (1.5,3) {\small$ ~~x=y $} ;
\node[above] at (1.2,1) {\small$ \ud $} ;
\draw[fill=black] (1,1) circle (0.04);
\fill[color=gray,opacity=0.5]  (0,0) -- (1,3) -- (1,1) -- (4.5,1) -- cycle;
\node at (4,4.6) {$ \mC $} ;
\node[left] at (0,7) {$ x_2 / y_2 $} ;
\node[right] at (7,0) {$ x_1 / y_1 $} ;
%\fill[gray!20,nearly transparent] (0,11) -- (0,12) -- (3,12) -- (3,8) -- cycle;
%\node at (4.5,5.2) {\small$ \mathcal{R}_{l} $} ;
%\node at (5.2,4.5) {\small$ \mathcal{R}_{l^\prime} $} ;
\end{tikzpicture}
%\end{center}
\caption{An illustration of the allocation for $ N = 2 $} \label{figone}
\end{figure}
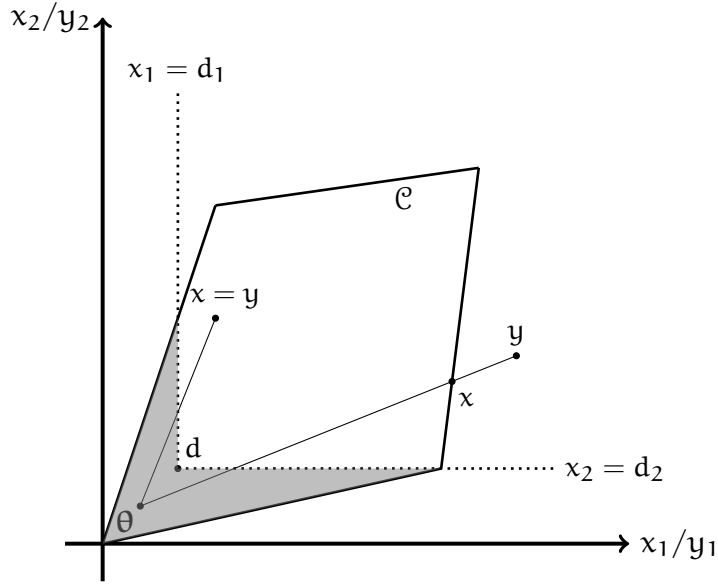
For any demand $ y \in \mathcal{S}_y \triangleq  \times_{i=1}^{N} \: (d_i,+\infty)  $, the allocation $ x $ will be equal to $ y $ if $ y $ is inside the feasibility region $ \mC $. Otherwise the allocation will be the intersection point between the boundary of the feasibility region $ \mC $ and the line joining $ \utheta \text{ with } y $ (that intersection point which lies between $ \utheta \text{ and } y $). The figure shows two different possible $ y $'s and their corresponding allocation $ x $. The shaded region represents that part of $ \mC $ that can never be allocated. Note that since $ \Vert \ud \Vert , \Vert \utheta \Vert \approx 0 $ this is a very small region and thus it doesn't significantly affect the generality of results presented in this paper. Also this can be seen as a partial justification for why assumption (A2) was needed. %For the cases involving degeneracy (i.e. equality constraints) the same proportional allocation idea can be carried through to the space $ \mathbb{R}_+^{N-K} $ and the feasibility set $ \widetilde{\mC} \subset \mathbb{R}_+^{N-K} $.

%Divide the demand set $ \mathcal{S}_y =  \times_{i=1}^{N} \: [d_i,+\infty)  $ into non-intersecting regions $ \left( \mathcal{R}_l \right)_{l \in \widetilde{\mL}} $ such that $ \mathcal{R}_l \triangleq \mathcal{S}_y \cap \widetilde{\mathcal{R}}_l $. Each region $\widetilde{\mathcal{R}}_l $ is the cone (with flat faces) formed by the vertex at $ \theta $ and rays joining $ \theta $ with the vertices of the polytope $ \mC $ that are also contained in the hyperplane $ \mF_l = \{ A_l^\top x = c_l \} $. Here we also consider the additional hyperplanes of the form $ \{x_i = D\} $, hence regions $ \mathcal{R}_l $ are defined for $ l \in \widetilde{\mL} \supseteq \mL $. (Note that since $ \theta \notin \mathcal{S}_y $ and $ \theta,d \in \text{int}(\mC) $, the cones $ \widetilde{\mathcal{R}}_l $ will indeed cover all points in $ \mC $.)

%The idea is to allocate rates that are always within the feasible set $ \mC $ (irrespective of the demand $ y $).

Formally, in case of feasibility set $ \mC $ not having any degeneracy, for a demand $ y \in \mathcal{S}_y $, the allocation $ x $ generated by the contract is
\begin{gather} \label{EQall}
x = \begin{cases}
y &\text{if}~~ y \in \mC \\
y_0 &\text{if}~~ y \notin \mC,
\end{cases}
%\\
%P(\alpha) \triangleq c \alpha^\gamma \exp \left( -k \alpha \right) ~~\text{for}~~ \alpha \ge 0
\end{gather}
where $ y_0 $ is the point on the boundary of $ \mC $ which is also on the line joining $ \utheta \text{ with } y $. Explicitly, if the above intersection happens on the hyperplane $ \mF_l = \{A_l^\top x = c_l \} $ then we can express $ y_0 $ as
\begin{equation}
y_0 - \theta = \alpha_0 (y-\theta)
\quad \text{with} \quad
\alpha_0 = \frac{c_l - A_l^\top \theta}{A_l^\top \left( y - \theta \right)}
\end{equation}
The above allocation mapping is an extension of generalized proportional allocation idea (see \cite{SiAn_multicast_arxiv} and \cite{SiAn_unicast_arxiv}), but modified in accordance with the generality of the problem~\eqref{CPg} and also so that points in the interior of $ \mC $ are covered as well. It is easy to verify in above that for $ y \notin \mC $, because of the way proportional allocation is defined, the expression for $ \alpha_0 $ above is well-defined and positive. Also, if one extends the definition of $ y_0 $ to the boundary i.e $ y  \in \partial \mC $ then it is easy to see that $ y_0 = y $ and $ \alpha_0 = 1 $.

Another useful explicit way of defining $ \alpha_0 $ is as follows.
\begin{gather}
\alpha_0 = \min_{\substack{l \in \mL \\ \alpha_0^l > 0}}~ \alpha_0^l
\quad \text{with} \quad \alpha_0^l = \frac{c_l - A_l^\top \theta}{A_l^\top \left( y - \theta \right)} ~~\forall~l \in \mL
\end{gather}
%In above, $ c_l - A_l^\top \theta > 0 $ for all $ l \in \mL $ (because $ \utheta \in \textrm{int}(\mC) $), and $ A_l^\top \left( y- \theta \right) $ is positive since for any $ i $, $ y_i > d_i > \theta_i $. The purpose of taking minimum of $ \alpha_0^l $ is to find the ``innermost'' (or closest) constraint to $ \utheta $ in the (positive) direction of $ y-\theta $.
In the above, $ c_l - A_l^\top \theta > 0 $ for all $ l \in \mL $ (because $ \utheta \in \textrm{int}(\mC) $), but $ A_l^\top \left( y- \theta \right) $ might be positive, negative or zero (since we are considering all $ l \in \mL $ and not restricting ourselves to appropriate regions as before). The purpose of taking minimum of $ \alpha_0^l $ over positive ones is to find the ``innermost'' (or closest) constraint to $ \utheta $ in the (positive) direction of $ y-\theta $.

Now we turn our attention to the case of degenerate feasible set. Since we are mainly interested in dealing with problems of (local) public goods nature,  we define an alternate characterization of the constraint set only for those cases\footnote{This generalizes in a straightforward way.}. We can rewrite constraints from (C$ _2 $) to explicitly account for equalities. For this consider disjoint sets of agents $ (\mN_k)_{k \in \mK} $, where $ \mK = \{1,\ldots,K\} $, such that $ \mN = \sqcup_{k \in \mK} \mN_k $, and the partition $ (\mN_k)_{k \in \mK} $ of $ \mN $ serves the purpose of grouping agents by locality. % and $ \hat{\mN} $ is the set of agents who aren't part of any locality (or form their own respective localities).
Thus the rewritten constraints are
%If there are (effectively) any equality constraints in~\eqref{CPg} then we can write them separately, so that (C$ _2 $) can be rewritten
%and set $ \hat{\mN} = \mN  \backslash (\sqcup_{k \in \mK} \mN_k) $. as
%\todo{Note the $ 0 $ in RHS of C$_{2a}$}
%\todo{I can make C$_{2a}$ more explicit by just assuming $ \tilde{x} = E x $}
\begin{subequations}
\begin{gather}
%E_k^\top x = 0 ~~\forall~k \in \mK \tag{C$ _{2a} $}\\
x_i = x_j ~~\forall~i,j \in \mN_k,~\forall~k \in \mK \tag{C$ _{2a} $}\\
%= (\widetilde{x},\hat{x}) ~~\text{with}~~ \hat{x} = E \widetilde{x} ~~\text{and}~~\widetilde{x} \in \mathbb{R}_+^{N-K} \\
A_l^\top x \le c_l ~~\forall~ l \in \widetilde{\mL} \subseteq \mL,  \tag{C$ _{2b} $}
\end{gather}
\end{subequations}
where $ \widetilde{\mL} $ is the subset of original constraints - the ones remaining after equality constraints have been separated. Without stating formally, we assume that the constraints in~(C$ _{2b} $) do not introduce any further degeneracy. %where $ K < N $  is the total number of independent (and consistent) equality constraints (so that dimensionality of $ x $ reduces by exactly $ K $).
Now one can consider a further refinement of above to state the constraints only in terms of free variables and thus containing only inequality constraints. For this define $ \widetilde{x} = (x_{j_k})_{k \in \mK} $ where $ j_k $ represents the agent with lowest index from $ \mN_k $ (to make the representation unique). Clearly this is exactly the set of free variables from above and we can rewrite (C$ _2 $) as
\begin{equation}
\widetilde{A}_l^\top \widetilde{x} \le c_l ~~\forall~l \in \widetilde{\mL}, \tag{C$ _{2c} $}
\end{equation}
with $ \widetilde{A}_l \in \mathbb{R}^K $ being derived from $ A_l $ by using (C$ _{2a} $) explicitly. Recall that from assumption (A4) we had at least 2 agents on each constraint, in cases such as above we extend that to have at least 2 agents per constraint in the above definition. %Although there are other equivalent ways to choose $ j_k $, the representation
%represents the free component in $ x $ and $ \hat{x} $ represents the dependent component.
%we assume that in this reformulation, constraints in~(C$ _{2a} $) are linearly independent and reduce the dimensionality of the feasibility space by exactly $ K $. This gives $ \widetilde{x} \in \mathbb{R}_+^{N-K} $ as the vector of free variables from $ x $. We can write this as $ \widetilde{x} = E x $ for $ E \in \mathbb{R}_+^{(N-K)\times N} $ - note that we assume entries of $ E $ as non-negative.
We denote by $ \widetilde{\mC} \subset \mathbb{R}_+^{K} $ the effective polytope in the $ \widetilde{x}- $space, created by the constraints in~(C$ _{2c} $).
%are written in terms of $ \widetilde{x} \in \mathbb{R}_+^{N-K} $, they form a non-degenerate polytope in the first quadrant of $ \mathbb{R}^{N-K} $ containing the origin i.e. there are no sources of degeneracy arizing from~(C$ _{2b} $). Also that this polytope can be described as
%\begin{gather}
%\widetilde{A}_l^\top \widetilde{x} \le \widetilde{c}_l ~~\forall~ l \in \widetilde{\mL}  \tag{C$ _{2c} $}
%\end{gather}
%with $ \widetilde{A}_l \in \mathbb{R}_+^{N-K} $ and $ \widetilde{c}_l \ge 0 $. Denote this polytope with $ \widetilde{\mC} $.
%\begin{gather}
%x = Hy + J \tag{C$ _{2c} $}\\
%\widetilde{A}_l^\top y \le \widetilde{c}_l ~~\forall~ l \in \widetilde{\mL}  \tag{C$ _{2d} $}
%\end{gather}
%for variable $ y \in \mathbb{R}_+^{N-K} $ and fixed $ H \in \mathbb{R}^{N\times(N-K)}, J \in \mathbb{R}^{N} $.

%We will use the alternate formulation for describing the allocation in Section~\ref{secmech} and as mentioned previously, it is to take care of degeneracy.

Coming back to allocation, for demand $ y \in \mathcal{S}_y $ we first create modified demand $ \widetilde{y} \in \times_{i=1}^{K} (\widetilde{d}_i,+\infty) \subset \mathbb{R}_+^{K} $ as follows:
\begin{equation} \label{EQytilde}
\widetilde{y}_k = \frac{1}{N_k} \sum_{i \in \mN_k} y_i ~~\forall~k \in \mK,
\end{equation}
i.e. averaging quoted $ y $'s from each group to get one representative demand.
%Recalling the alternate form (C$ _{2a} $), this linear map is as follows: for agents not involved in  (C$ _{2a} $) their demand is taken into $ \widetilde{y} $ as it is. For agents involved in (C$ _{2a} $), the free variables arizing out of the system are created by averaging over all possible permutations....
%\todo[inline]{I believe a permutation method of decription might work, but I am yet to find anything concrete.}
Then proportional allocation, as above, is performed on $ \widetilde{y} $ and the feasibility set $ \widetilde{\mC} $ to get $ \widetilde{x} \in \widetilde{\mC} $ (note that $ \widetilde{\utheta} $ - the restriction of $ \utheta $ to appropriate coordinates, will be used in place of $ \utheta $). This proxy allocation $ \widetilde{x} $ is finally converted back, using (C$ _{2a} $), to get $ x \in \mC $. % using $ x = (\widetilde{x},\hat{x}) $ and $ \hat{x} = E \widetilde{x} $. %the linear map arizing from the alternate form in~(C$ _{2a} $).
%\todo{I couldn't make it formal.}

%Before moving onto taxes we will go over the proofs of feasibility and stationarity.

%\hrule

\subsection{Taxes}

For any agent $ i $, we define his tax $ t_i $ as
\begin{subequations}
\begin{gather}
\label{EQtaxl}
t_i = \sum_{l \in \mL_i} t_i^l \\
\label{EQtaxg}
t_i^l = A_{li} x_i \bar{p}_{-i}^l + \left( p_i^l - \bar{p}_{-i}^l\right)^2 + \eta \, \bar{p}_{-i}^l p_i^l  \left( c_l - A_l^\top x \right)^2 \\ %+ g_i^l(s_{-i}) \\
\label{EQpbar}
\bar{p}_{-i}^l \triangleq \frac{1}{N^l-1}\sum_{\substack{j \in \mN^l \\ j \ne i}} p_j^l
\end{gather}
\end{subequations}
With $ \eta > 0 $ being a positive constant.

In the results section we will see how the various terms in the tax above represent different KKT conditions. Note for example the price-taking nature that $ A_{li}x_i\bar{p}_{-i}^l $ induces; here the agent is supposed to pay for his allocation but the price at which he has to pay is decided by others.

%\subsection{Interpretation of Mechanism and its applications}

\section{Equilibrium Results} \label{secres}

In this section we state and prove lemmas that will give us the desired full implementation property for the mechanism defined in Section~\ref{secmech}. Later in this section, the presented mechanism will be modified slightly to get an additional property of SBB at NE.

%\begin{theorem}[Full Implementation] \label{theomain}
%All Nash equilibria $ s = (y,P)  $ of the game $ \mathfrak{G} $ have the same corresponding allocation $ x $. Furthermore, this allocation is identical to the unique solution of~\eqref{CPg} i.e. $ x=x^\star $. Also, all agents are weakly better-off at equilibrium than by not participating at all.
%\end{theorem}

The main result of full implementation would require us to prove that all pure strategy NE of game $ \mathfrak{G} $ result in allocation $ x^\star $ - the unique solution of~\eqref{CPg}, and to show individual rationality. The method for proving this result is as follows: firstly we show (Lemmas~\ref{lempf}-\ref{lemst}) that for any pure strategy NE the corresponding allocation and quoted prices must satisfy the KKT conditions from Section~\ref{subseckkt}. Since by assumptions KKT conditions are both necessary and sufficient, this would mean that if pure strategy NE exist (unique or multiple), the corresponding allocation would have to be the solution of~\eqref{CPg} with quoted prices as the optimal Lagrange multipliers. Then we show existence (in Lemma~\ref{lemexis}) by explicitly writing out the messages that achieve the solution of~\eqref{CPg} and showing that there are no profitable unilateral deviations from those. Finally individual rationality will be shown (in Lemma~\ref{lemir}).

For Lemma~\ref{lemst}, we would need to distinguish between various cases arising from the general form of~\eqref{CPg}. We will make specific assumptions when necessary.
%\todo[inline]{Not sure about above}

\begin{lemma}[\textbf{Primal Feasibility}] \label{lempf}
For any message $ s = (y,P) \in \mathcal{S} $ with corresponding allocation $ x $, we have $ x \in \mC $ i.e. allocation is always feasible.
\end{lemma}
\begin{proof}
%Notice that the allocation $ w $ was created in a radial manner (around $ \theta $). It is easy to verify that for $ y \in \mathcal{R}_l $ we also have $ x \in \mathcal{R}_l $.
%
Consider the case where $ \mC $ doesn't consist of equality constraints i.e. is not degenerate. For $ y \in \mC $, feasibility of allocation $ x $ is obvious. For $ y \notin \mC $, the allocation is chosen on the boundary of $ \mC $, hence it is feasible as well (note that $ \mC $ is closed).

In case when $ \mC $ is degenerate, the dummy allocation $ \widetilde{x} $ is feasible w.r.t. $ \widetilde{\mC} $ using the above argument. But since the allocation $ x $ is produced from $ \widetilde{x} $ whilst satisfying the equality constraints, this means that $ x $ satisfies both the inequality and equality constraints and hence is feasible.%in a manner such that it is on the hyperplane $ \mF_l $, but recall that the regions $ \mathcal{R}_l $ were defined such that the constraint ``closest'' to $ \theta $ would be $ \mF_l $. This gives that the allocation is feasible in this case as well.
%Take $ z = (y,n(y)) \in \mathcal{R}_l $ for any $ l $. It is easy to see from above that $ w \in \mathcal{R}_l $ as well.
%\begin{equation}
%\Vert w-\theta \Vert = \vert P(\alpha_0) \vert \cdot \Vert z_0 - \theta \Vert \le \Vert z_0 - \theta \Vert
%\end{equation}
%Hence $ w $ is closer to $ \theta $ than $ z_0 $ and is thus inside the feasible set $ \mC $.
%Here is an alternate way to prove feasibility. Since $ w \in \mathcal{R}_l $, all we need to verify is $ A_l^\top x + B_l^\top m \le c_l $. Using~\eqref{EQall}
%\begin{gather}
%A_l^\top x + B_l^\top m = A_l^\top \left((x_0 - \theta_x)P(\alpha_0) + \theta_x \right) + B_l^\top ((m_0 - \theta_m)P(\alpha_0) + \theta_m) \\
%= P(\alpha_0)  \left( A_l^\top x_0 + B_l^\top m_0 \right)  + (1-P(\alpha_0)) \left( A_l^\top \theta_x + B_l^\top \theta_m \right)
%\end{gather}
%Since $ z_0 $ is on the hyperplane $ \mF_l $, we have $ A_l^\top x_0 + B_l^\top m_0 = c_l $. Moreover $ \theta $ is inside the feasible set, so it satisfies the the $ l\textsuperscript{th} $ constraint $ \Rightarrow $ $ A_l^\top \theta_x + B_l^\top \theta_m < c_l $. So the above weighted average is also less than $ c_l $.
\end{proof}

\begin{lemma}[\textbf{Equal Prices}] \label{lemep}
At any NE $ s = (y,P) $ of the game $ \mathfrak{G} $, $ \forall $ $ l \in \mL $ and $ \forall $ $ i \in \mN^l $, we have $ p_i^l = p^l $ i.e. all agents on a constraint quote the same price for that constraint.
\end{lemma}
\begin{proof}
Suppose there exists a constraint $ l \in \mL $ such that $ (p_j^l)_{j \in \mN^l} $ are not all equal. Then there exists an agent $ i \in \mN^l $ such that $ p_i^l > \bar{p}_{-i}^l $ (this is clear from definition in~\eqref{EQpbar}).

Consider the downwards price deviation $ {p_i^l}^\prime = \bar{p}_{-i}^l \ge 0 $ for this agent (whilst keeping all other values in $ s_i $ the same). The difference in utility will be
\begin{subequations}
\begin{gather}
\Delta \widehat{u}_i = \widehat{u}_i(s_i^\prime,s_{-i}) - \widehat{u}_i(s_i,s_{-i}) = -({p_i^l}^\prime - \bar{p}_{-i}^l)^2 + (p_i^l - \bar{p}_{-i}^l)^2 \nonumber \\
- {} \etas \bar{p}_{-i}^l {p_i^l}^\prime \left( c_l - A_l^\top x \right)^2 + \etas \bar{p}_{-i}^l p_i^l  \left( c_l - A_l^\top x \right)^2 \\
=  \underbrace{(p_i^l - \bar{p}_{-i}^l)^2}_{> 0} \, + \,\etas \bar{p}_{-i}^l (\underbrace{ p_i^l - \bar{p}_{i}^l }_{> 0}) \left( c_l - A_l^\top x \right)^2 > 0
\end{gather}
\end{subequations}
Hence such a unilateral deviation is profitable. Therefore at NE, all quoted prices for any constraint have to be the same.
\end{proof}

Since we have shown that there is a common price $ p^l $ for any constraint $ l \in \mL $, here onwards we can refer only to $ p^l $'s instead of $ p_i^l $'s. Here onwards, whenever we refer to quoted price $ P $, it is to be interpreted as $ (p^l)_{l \in \mL} $.

\begin{lemma}[\textbf{Dual Feasibility}] \label{lemdf}
For any $ l \in \mL $, $ p^l \ge 0 $.
\end{lemma}
\begin{proof}
This is obvious since quoted prices are all non-negative.
\end{proof}

\begin{lemma}[\textbf{Complimentary Slackness}] \label{lemcs}
For any NE $ s = (y,P) $ of game $ \mathfrak{G} $ with corresponding allocation $ x $, we have for any $ l \in \mL $
\begin{equation}
p^l \left( c_l - A_l^\top x \right) = 0
\end{equation}
\end{lemma}
\begin{proof}
Suppose there exists a constraint $ l \in \mL $ such that, at NE, $ p^l > 0 $ and $ c^l > A_l^\top x $. We will take a downwards price deviation $ {p_i^l}^\prime = p^l - \varepsilon > 0 $ for any agent $ i \in \mN^l $. Again we can write the difference in utility as
\begin{subequations}
\begin{gather}
\Delta \widehat{u}_i = -({p_i^l}^\prime - \bar{p}_{-i}^l)^2 + 0 - \etas \bar{p}_{-i}^l {p_i^l}^\prime \left( c_l - A_l^\top x \right)^2 + \etas \bar{p}_{-i}^l p_i^l \left( c_l - A_l^\top x \right)^2  \\
= -(-\varepsilon)^2  + \etas p^l \left(\varepsilon \right) \left( c_l - A_l^\top x \right)^2 = - \varepsilon \left( \varepsilon - a \right)
\end{gather}
\end{subequations}
with $ a = \etas p^l \left( c_l - A_l^\top x \right)^2 > 0 $. Now choosing $ 0 < \varepsilon < \min(a,p^l) $ this can be made into a profitable deviation.
\end{proof}

Next we have the lemma for the stationarity property from KKT. For this we have to make assumptions on the feasible set $ \mC $.
\begin{itemize}
\item[(A6)] \noindent\fbox{\parbox{0.9\textwidth}{
\begin{equation}
\widetilde{A}_l \in \mathbb{R}_+^K ~~\forall~ l \in \widetilde{\mL}
\end{equation}
where $ \widetilde{A}_l $ was defined in (C$ _{2c} $) (Section~\ref{secalo}). }}
\end{itemize}
The above clearly translates into $ A_l \in \mathbb{R}_+^N ~~\forall~ l \in {\mL} $ in case there is no degeneracy. But when there is degeneracy, the above assumption states that the effective polytope $ \widetilde{\mC} $ in the $ \widetilde{x}- $space has faces whose outward normals are pointing away from the origin.
%\todo{Some more comments?}

\begin{lemma}[\textbf{Stationarity}] \label{lemst}
At any NE $ s = (y,P) $ of game $ \mathfrak{G} $ with corresponding allocation $ x $, we have for any agent $ i \in \mN $,
\begin{equation} \label{EQstat1}
v_i^\prime(x_i) = \sum_{l \in \mL_i} A_{li} p^l
\end{equation}
%and for any $ j \in \mN $
%\begin{equation} \label{EQstat2}
%\sum_{l \in \mL_j} B_{lj} p^l = 0.
%\end{equation}
\end{lemma}
\begin{proof}
%We start with~\eqref{EQstat2}, \ldots
%
For any agent $ i \in \mN $ we can write
\begin{equation}
\frac{\partial \widehat{u}_i(s)}{\partial y_i} = \left( v_i^\prime(x_i) - \frac{\partial t_i}{\partial x_i} \right) \frac{\partial x_i}{\partial y_i}
\end{equation}
Note that $ \beta \equiv \partial x_i/\partial y_i $ isn't always defined, since the allocation is continuous but only piecewise differentiable for $ y \notin \mC $. But right and left derivatives are always defined. Noting that $ v_i^\prime(x_i) - \frac{\partial t_i}{\partial x_i} = 0 $ is equivalent to~\eqref{EQstat1} (using previous lemmas characterizing NE), it will be sufficient for us if we show that $ \beta > 0 $. Since then without making $ v_i^\prime(x_i) - \frac{\partial t_i}{\partial x_i} = 0 $ there is always an upwards or downwards deviation in $ y_i $ to make agent $ i $ strictly better-off.  %then we can easily see from above that we get the stationarity property in~\eqref{EQstat1} at NE (using Lemma~\ref{lemcs} to make other terms zero in the tax derivative). So we will show that even for $ \beta = 0 $ there exists a deviation which will make agent $ i $ strictly better-off, unless~\eqref{EQstat1} holds.

For $ y \in \mC $ we have $ x_i = y_i $, therefore $ \beta = 1 \ne 0 $. Otherwise, first consider the case where $ \mC $ is non-degenerate. Then
\begin{subequations}
\begin{gather}
x_i = \theta_{i} + \alpha_0 \left(y_i - \theta_{i}\right)
\quad \text{with} \quad
\alpha_0 = \frac{c_l - A_l^\top \theta}{A_l^\top \left( y - \theta \right)}
\\
\frac{\partial x_i}{\partial y_i} = \alpha_0 + \left(y_i - \theta_{i}\right) \frac{\partial \alpha_0}{\partial y_i}
\quad \text{with} \quad
\frac{\partial \alpha_0}{\partial y_i} = \frac{-\alpha_0 A_{li}}{A_l^\top \left( y - \theta \right)} \\
\Rightarrow~~\frac{\partial x_i}{\partial y_i} = \alpha_0 \left( 1 - \frac{A_{li}\left(y_i - \theta_{i}\right)}{A_l^\top \left( y - \theta \right)} \right) = \alpha_0 \Big( \sum_{j \in \mN^l \backslash \{i\}} A_{lj}\left(y_j - \theta_{j}\right) \Big)
\end{gather}
\end{subequations}
Noting that $ \alpha_0 > 0 $ and that due to assumption (A4) there are always at least two agents on any constraint, clearly the above is positive (since $ \theta_j < d_j < y_j $ $ \forall~j $).

In the case when $ \mC $ does have equality constraints, note that allocations are created by composing the maps $ y \xmapsto{A} \widetilde{y} \xmapsto{B} \widetilde{x} \xmapsto{C} x $ where $ B $ is proportional allocation
%\footnote{Note that (A6) is sufficient for $ \widetilde{A}_l $ (defined in (C$ _{2c} $)) to have non-negative entries}
on the set $ \widetilde{\mC} $ and $ A,C $ are linear with positive coefficients\footnote{Since they consist of averaging,~\eqref{EQytilde}, and assigning same value to multiple positions.}. Thus using the fact that for proportional allocation we have $ \beta > 0 $ (from above) and that $ A,C $ have positive coefficients, we clearly have $ \beta > 0 $ here as well.
\end{proof}

This completes the necessary part of our proof. The argument in the previous Lemmas applied to pure strategy NE if there exited any. Now we will prove the existence of pure strategy NE which will give the optimal allocation $ x^\star $.

\begin{lemma}[\textbf{Existence}] \label{lemexis}
For game $ \mathfrak{G} $, there exists NE $ s = (y,P) \in \mathcal{S} $ such that the corresponding allocation $ x $ and prices $ (p^l)_{l \in \mL} $ satisfy the KKT conditions as $ x^\star $ and $ (\lambda_l^\star)_{l \in \mL} $, respectively.
\end{lemma}
\begin{proof}
This proof will be done in two separate parts: first we will show that there exists $ s \in \mathcal{S} $ such that the allocation through the contract is $ x^\star $ and prices are Lagrange multipliers $ (\lambda^\star_l)_{l \in \mL} $. Secondly, we will show that for all claimed NE points, there doesn't exist a unilateral deviation that is profitable.

Existence of prices that match Lagrange multipliers $ \lambda_l^\star $ is obvious, since agents can quote any price in $ \mathbb{R}_+ $ and dual feasibility says that $ \lambda_l^\star \ge 0 $. For allocation to be the same as $ x^\star $, notice that due to assumption (A2) all possible solutions $ x^\star $ lie in the set $ \mathcal{Y} \triangleq \mC \cap \mathcal{S}_y $ with $ \mathcal{S}_y = \times_{i=1}^{N} \: (d_i,+\infty) $. Quoted demand $ y $ can be anywhere in the set $ \mathcal{S}_y $ and for $ y \in \mC $ (i.e. $ y \in \mathcal{Y} $) the allocation is $ x = y $ and therefore each point $ x \in \mathcal{Y} $ is achievable as allocation by quoting the same point $ y = x \in \mathcal{Y} $. (However also note that points on the boundary of $ \mathcal{Y} $ are achievable as allocation by many $ y $'s outside $ \mathcal{Y} $ as well.)

Now we will check for unilateral profitable deviations. When quoted demand creates allocation as $ x^\star $ and the quoted prices are equal and equal to $ \lambda^\star $, the utility for any agent $ i $ is
\begin{equation}
\widehat{u}_i(s) = v_i(x_i^\star) - x_i^\star \sum_{l \in \mL_i} A_{li}\lambda_l^\star
\end{equation}
Due to the fact that $ f(x) = v_i(x) - x \sum_{l \in \mL_i} A_{li}\lambda_l^\star $ is strictly concave and $ f^\prime(x_i^\star) = 0 $ (Stationarity) we can conclude that $ x_i^\star $ is a global maximizer of $ f $. With this we have
\begin{subequations}
\begin{gather}
\widehat{u}_i(s) = v_i(x_i^\star) - x_i^\star \sum_{l \in \mL_i} A_{li} \lambda_l^\star \ge v_i(x_i) - x_i \sum_{l \in \mL_i} A_{li} \lambda_l^\star \\
\label{EQdev}
\ge v_i(x_i) - x_i \sum_{l \in \mL_i} A_{li} \lambda_l^\star - \sum_{l \in \mL_i} \left(p_i^l - \lambda_l^\star \right)^2 - \sum_{l \in \mL_i} \etas p_i^l \lambda_l^\star \left(c_l - A_l^\top x \right)^2
\end{gather}
\end{subequations}
The above is true for any $ \left(x_j\right)_{j \in \mN} $ and $ \left(p_i^l\right)_{l \in \mL_i} $ non-negative and final inequality holds because the additional terms are non-positive. Now any unilateral deviation $ (s_i^\prime, s_{-i}) $ from agent $ i $ will result in utility for agent $ i $ which has the form as in~\eqref{EQdev}. Hence we have proved that with unilateral deviations from messages that have allocation $ x^\star $ and prices $ \lambda^\star $, the corresponding agent can never be strictly better off.
\end{proof}

\begin{lemma}[\textbf{Individual Rationality}] \label{lemir}
For any NE $ s = (y,P) $ of game $ \mathfrak{G} $ with corresponding allocation $ x $ and taxes $ t $, we have for any $ i \in \mN $
\begin{equation}
u_i(x,t) \ge u_i(0,0) %\quad \forall~i \in \mN
\end{equation}
\end{lemma}
\begin{proof}
Recall that $ u_i(x,t) = v_i(x_i) - t_i $. For any agent $ i \in \mN $, define
\begin{equation}
f(z) \triangleq v_i(z) - z \sum_{l \in \mL_i} A_{li} p^l
\end{equation}
Note that $ u_i(0,0) = v_i(0) = f(0) $ and at NE $ u_i(x,t) = v_i(x_i) - x_i \sum_{l \in \mL_i} A_{li} p^l = f(x_i)  $ (also recall that except the first term, all other tax terms go to zero at NE, refer~\eqref{EQtaxg}).

By stationarity property we know $ f^\prime(x_i) = v_i^\prime(x_i) - \sum_{l \in \mL_i} A_{li} p^l = 0 $ and by assumption (A1) it is clear that $ f(\cdot) $ is strictly concave. Therefore $ f^\prime(y) > 0 $ for $ 0 < y < x_i $ and we can claim by mean value theorem that $ f(x_i) > f(0) $.
\end{proof}

%\begin{lemma}[Strong Budget Balance] \label{lemsbb}
%content...
%\end{lemma}
%\begin{proof}
%content...
%\end{proof}

\begin{theorem}[\textbf{Full Implementation}] \label{theomainnew}
All Nash equilibria $ s = (y,P)  $ of the game $ \mathfrak{G} $ have the same corresponding allocation $ x $. Furthermore, this allocation is identical to the unique solution of~\eqref{CPg} i.e. $ x=x^\star $. Also, all agents are weakly better-off at equilibrium than by not participating at all.
\end{theorem}
\begin{proof} %[(Theorem~\ref{theomainnew})]
Thus from Lemma~\ref{lemexis} there exists at least one pure strategy NE. From Lemma~\ref{lempf}-\ref{lemst}, allocation $ x $ at any NE has to satisfy KKT conditions and knowing that KKT conditions are necessary and sufficient for optimality, it is clear that $ x = x^\star $. Finally, individual rationality was shown in Lemma~\ref{lemir} and with this we have our full implementation result.
\end{proof}

\subsection{Strong Budget Balance (at equilibrium)} \label{secsbboneq}

In this section we present a modification in the above mechanism so that in addition to above properties we also have SBB, i.e.,  $ \sum_{i \in \mN} t_i = 0 $ at those equilibria where $ y \in \mC $.

The message space and allocation $ x $ is exactly same as before, the only modification will be to the taxes. Whenever the quoted demand is feasible, we have $ x = y $. So from~\eqref{EQtaxl},~\eqref{EQtaxg} and previous lemmas we can write total tax at NE as
%\begin{subequations}
\begin{gather} \label{EQtaxredb}
\sum_{i \in \mN} t_i = \sum_{i \in \mN} \sum_{l \in \mL_i} t_i^l = \sum_{l \in \mL} \sum_{i \in \mN^l} t_i^l
= \sum_{l \in \mL} \sum_{i \in \mN^l}    A_{li} y_i p^l %+ \left(p_i^l - \bar{p}_{-i}^l\right)^2 + \etas \bar{p}_{-i}^l p_i^l \left(c_l - A_l^\top y \right)^2 \right)
\end{gather}
%\end{subequations}

To achieve SBB, our approach is to redistribute the above total tax amongst all the agents so that the amount of money received by an agent will not be controlled by that agent i.e. it will be only a function of $ y,p $ quoted by other agents. This will ensure that strategic decisions are still same as before and thus we would continue to have all the previous properties (Lemmas~\ref{lempf}-\ref{lemexis}). Furthermore, we will redistribute taxes per constraint. So for any constraint $ l \in \mL $, $ \sum_{i \in \mN^l} t_i^l $ will be redistributed only amongst agents in $ \mN^l $. %In our mechanism, for $ y \notin \mC $ the allocation $ x $ is created in a proportional manner, which means that each $ y_i $ controls every $ x_j $. Thus it is clear from~\eqref{EQtaxg} that this method of redistribution will only work for $ y \in \mC $.
%SBB at NE here refers to SBB at that equilibrium where $ y \in \mC $ (there always exists one).
After appropriately redistributing the taxes, we will check individual rationality at equilibrium once again since the total tax paid by an agent might have a different value now. The reason we present SBB as a side result is that one can always use such a redistribution technique and thus SBB generally reduces to appropriate algebraic manipulation and can usually be added later on in the mechanism (for problems of these types).%we will modify the tax component of the contract (see~\eqref{EQtaxg}) but only in a way so that

Let $ T_i^l $ be the new tax where analogously the old tax (from~\eqref{EQtaxg}) was $ t_i^l $. Then the total tax for any agent $ i $ is $ T_i = \sum_{l \in \mL_i} T_i^l $. From above discussion, what we want is
\begin{equation} \label{EQtaxnew}
T_i^l = t_i^l - f_i^l\left((y_j,p_j^l)_{\substack{j \in \mN^l \\ j \ne i}}\right)
\end{equation}
for any $ i \in \mN $ and $ l \in \mL_i $. The function $ f_i^l $ will be defined below. %for SBB at equilibrium and separately in Section~\ref{secsbboffeq} for SBB on and off-equilibrium.

%The results regarding SBB require stronger assumptions on~\eqref{CPg} than the ones in Section~\ref{secCP}. We will state these modified assumptions as they are needed. %, and start with modifying (A4).
%
%\subsection{SBB at equilibrium}
%
%For simplicity the SBB modification has been divided into two parts - in this part instead of redistributing the whole total tax from~\eqref{EQtaxredb}, we will only redistribute
%Since at equilibrium, due to Lemma~\ref{lemep} and \ref{lemcs}, only the first term in~\eqref{EQtaxredb} is non-zero, here we will be redistributing
%\begin{equation}
%\sum_{l \in \mL} \sum_{i \in \mN^l} A_{li} y_i \bar{p}_{-i}^l = \sum_{l \in \mL} \sum_{i \in \mN^l} A_{li} y_i p^l
%\end{equation}
%since we know that the other two terms are zero at NE.
%With this we will achieve $ \sum_{i \in \mN} t_i = 0 $ at NE. Also note from the previous discussion that for SBB to work out we consider only those equilibria where $ y \in \mC $. From the construction of allocation we have $ x = y $ for all $ y \in \mC $, hence there is always guaranteed to be an equilibrium for which $ y \in \mC $. In fact whenever the optimal allocation is such that $ x^\star \in \text{int}\left(\mC\right) $ there is only one possible $ y $ at NE (namely $ x $ itself).
%
%Before specifying $ f_i^l $, we will modify the assumption (A4).
%\begin{itemize}
%\item[(A4$ ^\prime $)] \noindent\fbox{\parbox{\textwidth}{For any $ l \in \mL $, $ N^l \ge 3 $.}}
%\end{itemize}

For $ l \in \mL $ such that $ N^l \ge 3 $, in case where $ A_{li} \ge 0 $ $ \forall $ $ i \in \mN^l $,
we have
\begin{equation} \label{EQf1sbb}
f_i^l \equiv f_{i,1}^l = \frac{1}{N^l - 2}\sum_{\substack{j \in \mN^l \\ j \ne i}} A_{lj} y_j \left(\bar{p}_{-i}^l - \frac{p_j^l}{N^l - 1} \right)
\end{equation}
when there are negative coefficients in $ A_{li} $, we can define $ f_i^l $ in the same manner as above, except that coefficients for those $ j $ that are involved in an equality constraint will be (scaled)\footnote{the scaling factor being the inverse of the number of agents on $ l $ that are involved in constraints that make their allocation equal.} $ \widetilde{A}_{lj} $ instead of $ A_{lj} $.

%For general constraints cnsider $ I^l_+ = \{A_\} $
For $ l \in \mL $ such that $ N^l = 2 $ (assuming $ \mN^l = \{i,j\} $), we divide the redistribution into 2 cases: first we consider the case $ A_{li}, A_{lj} \ge 0 $, for this we have
\begin{equation} \label{EQf1sbb2}
f_i^l \equiv f_{i,1}^l = A_{lj} y_j p_j^l.
\end{equation}
If the coefficients $ A_{li} $ are not positive then we are essentially into the degenerate case, where $ \vert A_{li} \vert = \vert A_{lj} \vert = 1 $ with opposite signs. For this we have $ f_i^l \equiv f_{i,1}^l = 0 $ and same for agent $ j $.
%\todo{Above is true, but not obvious}

With the modified tax we will denote the new game as $ \mathfrak{G}^\prime $.

\begin{lemma}[\textbf{Strong Budget Balance at NE}] \label{lemsbbeq}
For the game $ \mathfrak{G}^\prime $, at all NE $ s = (y,P) $ where $ y \in \mC $ we have
\begin{equation}
\sum_{i \in \mN} T_i = \sum_{i \in \mN} \sum_{l \in \mL_i} T_i^l = 0 \tag{SBB}
\end{equation}
\end{lemma}
\begin{proof}
Recall that due to Lemmas~\ref{lemep} and \ref{lemcs}, at NE we have
\begin{equation}
\sum_{i \in \mN} t_i = \sum_{i \in \mN} \sum_{l \in \mL_i} t_i^l = \sum_{l \in \mL} \sum_{i \in \mN^l} t_i^l = \sum_{l \in \mL} \sum_{i \in \mN^l} A_{li} y_i p^l
\end{equation}
For completing the proof, it will be sufficient to show $ \sum_{i \in \mN^l} t_i^l = \sum_{i \in \mN^l} f_i^l $ for any $ l \in \mL $, at NE. Let's begin with links with $ N^l = 2 $ (where $ \mN^l = \{i,j\} $). In case $ A_{li}, A_{lj} \ge 0 $ the total payment is
\begin{gather}
f_i^l + f_j^l = A_{lj} y_j p_j^l + A_{li} y_i p_i^l = \sum_{k \in \mN^l} A_{lk} y_k p^l.
\end{gather}
In the degenerate case, the tax paid by agents $ i,j $ at NE is $ y_i p^l, y_j p^l $ with opposite signs. But note that $ y_i = y_j $ at NE, since we are in the degenerate case. Hence the total tax is zero, which is the same as $ f_i^l + f_j^l $. Now for other links,
\begin{subequations}
\begin{gather}
\sum_{i \in \mN^l} \frac{1}{N^l - 2} \sum_{\substack{j \in \mN^l \\ j \ne i}} A_{jl} y_j \left(\bar{p}_{-i}^l - \frac{p_j^l}{N^l - 1} \right) = \frac{1}{N^l - 2} \sum_{i \in \mN^l} \sum_{\substack{j \in \mN^l \\ j \ne i}} A_{jl} y_j \left(\bar{p}_{-j}^l - \frac{p_i^l}{N^l - 1}\right)
\\
= \frac{N^l - 1}{N^l - 2} \sum_{j \in \mN^l} A_{jl} y_j \bar{p}_{-j}^l - \frac{1}{N^l - 2} \sum_{j \in \mN^l} \sum_{\substack{i \in \mN^l \\ i \ne j}} A_{jl} y_j \frac{p_i^l}{N^l - 1}
\\
= \frac{N^l - 1}{N^l - 2} \sum_{j \in \mN^l} A_{jl} y_j \bar{p}_{-j}^l - \frac{1}{N^l - 2} \sum_{j \in \mN^l} A_{jl} y_j \bar{p}_{-j}^l = \sum_{j \in \mN^l} A_{jl} y_j \bar{p}_{-j}^l = \sum_{j \in \mN^l} A_{jl} y_j p^l
\end{gather}
\end{subequations}
(the proof when $ A_{lj} $ are possibly negative also follows from above after explicitly using that equality constraints would make some allocations equal to each other at NE)
\end{proof}

After this, we have to check individual rationality. In cases when there is no degeneracy in $ \mC $, due to assumption (A6), $ A_{li} \ge 0 $ $ \forall~l,i $. This means that with $ f_{i}^l $ we have redistributed taxes by paying agents back, thus the redistribution is indeed non-negative (can be seen explicitly from~\eqref{EQf1sbb} and~\eqref{EQf1sbb2}). Thus if the NE was individually rational without this redistribution, it will continue to be so now.

For the cases involving degeneracy, for links associated with equality constraints the redistribution is zero. So individual rationality isn't affected. For other links, individual rationality follows from the argument in the previous paragraph after noting that effectively $ \widetilde{A}_{li} $ were used in place of $ A_{li} $ for defining $ f_{i}^l $ (and by (A6), $ \widetilde{A}_{li} \ge 0 $).

\section{Off-equilibrium Results} \label{secoffeq}

In this section we discuss (and prove) additional off-equilibrium properties of the mechanism. The original requirements of full implementation are only restricted to equilibrium properties. However we believe that in a realistic scenario, off-equilibrium properties are essential to justify working within a Nash-implementation framework. Nash equilibrium is an appropriate solution concept for complete information games i.e. games where it is assumed that all agents have complete knowledge about each others utility functions and also each agent knows that others know about his utility functions and so on (infinite hierarchy of beliefs). We are proposing this mechanism to be used in an informationally, and possibly physically, decentralized network setting like the Internet. The root of delving into mechanism design is that the designer doesn't have information about the utility functions, so to assume that agents themselves have all the information is impractical, especially for communications settings. In absence of complete information with agents, Nash implementation has still been used in literature - the justification being online or offline learning. In this formulation, before playing the actual game, agents participate in a multi-round learning process. After the agents have learnt about each others utilities (or more specifically the equilibrium action) they play the actual one-shot game.

In order to make the above learning model practical, there have to be real incentives involved for learning correctly. It is expected that while the learning process is still going on, in every round agents will quote demands and prices and receive allocation and taxes through the contract. We are interested in dealing with constrained resource allocation problems where the constraints could possibly be of a hard nature i.e. impossible to violate. An instance of this is the capacity constraint in unicast and multi-rate multicast examples. Hence if one assumes the learning model, then a necessary property for the mechanism would be that allocation is always (even off-equilibrium) feasible, since during the learning process agents will be playing off-equilibrium. To a lesser extent, if one is interested in SBB then same would have to be true of SBB as well.
It is for these reasons that we introduced the proportional allocation, as a distinct feature of our mechanism,
as opposed to the simple allocation used in \cite{healy2012designing}, \cite{demoscorrection},
where allocation equals demand, i.e. $ x = y $ everywhere.

Since the mid-90s, various learning models in Mechanism design and Game theory have been extensively studied. Readers may refer to
\cite{healy2012designing} for studying how mechanism design and learning are handled together. Also one may refer to
\cite{young2004strategic}, \cite{fudenberglearning},
\cite{tembine2012distributed},
\cite{shoham2009multiagent},
\cite{networkgames},
\cite{lasaulce2011game} for a compendium of existing results regarding learning in a strategic set-up.
%and the general argument with the help of argument of focussing.

%\todo[inline]{Some more justification may be added. like invalid contract if one know that some promises are false} %Citations of focussing missing}

%\subsection{Feasibility off-equilibrium}
%\todo[inline]{some more editing to be done like adding learning references etc}
%\begin{comment}

Now we extend the modification from Section~\ref{secsbboneq} to achieve SBB off-equilibrium.

\subsection{SBB off-equilibrium} \label{secsbboffeq}

Here we will define $ f_i^l $ from~\eqref{EQtaxnew} differently so that we achieve the property of $ \sum_{i \in \mN} t_i = 0 $ at all points where $ y \in \mC $, not just at NE.
From the expressions in~\eqref{EQtaxredb} as well as the description of $ f_i^l $ below, it will be clear that this method of redistribution only works when there are sufficient number of agents on every link. So here we will %further
modify assumption (A4) so that
\begin{itemize}
\item[(A4$ ^{\prime} $)] \noindent\fbox{\parbox{0.9\textwidth}{For any $ l \in \mL $, $ N^l \ge 5 $ i.e. there are at least $ 5 $ agents on any constraint.}}
\end{itemize}
Also, we will only deal with the non-degenerate case where $ A_{li} \ge 0 $ $ \forall $ $ l,i $. The corresponding results regarding degenerate cases would involve tedious case-by-case analysis and unnecessarily complicate the analysis of what should be a straightforward addition to the mechanism.

%Note that we are still operating under assumptions (A1$ ^\prime $) and (A4b).

Now $ f_i^l $ is defined in three parts: $ f_i^l\left((y_j,p_j^l)_{\substack{j \in \mN^l \\ j \ne i}}\right) = f_{i,1}^l + f_{i,2}^l + f_{i,3}^l  $. The three terms here are individually redistributing the three tax terms from~\eqref{EQtaxredb}. % and $ f_{i,1}^l $ is exactly the same as in~\eqref{EQf1sbb}.
\begin{subequations}
\begin{gather}
\label{EQfi1}
f_i^l \equiv f_{i,1}^l = \frac{1}{N^l - 2}\sum_{\substack{j \in \mN^l \\ j \ne i}} A_{lj} y_j \left(\bar{p}_{-i}^l - \frac{p_j^l}{N^l - 1} \right) \\
\label{EQfi2}
f_{i,2}^l = \frac{N^l}{(N^l - 1)^2(N^l - 2)} \sum_{\substack{k > m \\ k,m \ne i}} \left(p_k^l - p_m^l\right)^2 \\
f_{i,3}^l = \eta \left( f_{i,3a}^l + f_{i,3b}^l + f_{i,3c}^l \right)
\\
f_{i,3a}^l = \frac{2c_l^2}{\left(N^l - 1\right)\left(N^l - 2\right)} \sum_{\substack{k > m \\ k,m \ne i}}  p_k^l p_m^l
\\
f_{i,3b}^l = \frac{2}{N^l - 1} \sum_{\substack{j > q \\ j,q \ne i}} \Big[ \frac{1}{N^l - 3}  \sum_{\substack{k \in \mN^l \\ k \ne i,j,q}} p_j^l p_q^l \phi(y_k) +  \frac{1}{N^l - 2} \sum_{\substack{k \in \mN^l \\ k = j,q}} p_j^l p_q^l \phi(y_k) \Big]
\end{gather}
\end{subequations}
where $ \phi(y_k) = A_{kl}^2 y_k^2 - 2 c_l A_{kl} y_k $. For $ f_{i,3c}^l $ define the set $ B \triangleq \{j,q\} \cap \{k,s\} $,
\begin{gather}
f_{i,3c}^l = \frac{4}{N^l - 1} \sum_{\substack{j > q \\ j,q \ne i}} \Bigg[ \frac{1}{N^l-4}  \sum_{\substack{k > s \\ k,s \ne i \\ \vert B \vert = 0}} \psi + \frac{1}{N^l-3}  \sum_{\substack{k > s \\ k,s \ne i \\ \vert B \vert = 1}}  \psi + \frac{1}{N^l-2}  \sum_{\substack{k > s \\ k,s \ne i \\ k = j, s = q}} \psi \Bigg]
\end{gather}
where $ \psi = \psi(p_j^l,p_q^l,y_k,y_s) =   p_j^l p_q^l \left( A_{kl} y_k \right) \left( A_{sl} y_s \right) $. It is this last expression that necessitates the assumption (A4$ ^{\prime} $). We will denote the game with the above modified taxes with $ \mathfrak{G}^{\prime\prime} $.
%From the above expressions it is clear why in general we would require $ N^l \ge 5 $. For $ \mN^l = \left\{1,2,3\right\} $ we have
%\begin{gather}
%f_1^l(y_2,p_2^l,y_3,p_3^l) = \frac{1}{2} \left(y_2 p_3^l + y_3 p_2^l\right) + \left(p_2^l - p_3^l\right)^2
%\end{gather}

\begin{lemma}[Strong Budget Balance on and off-NE] \label{lemsbb}
For the game $ \mathfrak{G}^{\prime\prime} $, for all points in the message space $ \mathcal{S} $ where $ y \in \mC $
\begin{equation}
\sum_{i \in \mN} T_i = \sum_{i \in \mN} \sum_{l \in \mL_i} T_i^l = 0 \tag{SBB}
\end{equation}
\end{lemma}
\begin{proof}
By rearranging the above sum: $ \sum_{i \in \mN} \sum_{l \in \mL_i} T_i^l = \sum_{l \in \mL} \sum_{i \in \mN^l} T_i^l $, we can see that it will be sufficient if we show that for any constraint $ l \in \mL $,
%\begin{subequations}
\begin{equation} \label{EQtaxred2}
\sum_{i \in \mN^l} T_i^l = 0
\quad \Leftrightarrow \quad
\sum_{i \in \mN^l} t_i^l = \sum_{i \in \mN^l} f_i^l\left((y_j,p_j^l)_{\substack{j \in \mN^l \\ j \ne i}}\right)
\end{equation}
%\end{subequations}
We will again show this term-by-term; we have already shown in the proof of Lemma~\ref{lemsbbeq} that the sum of $ f_{i,1}^l $ is equal to the sum of payment (first) term  from $ t_i^l $, similarly the sum of $ f_{i,2}^l $ and $ f_{i,3}^l $ will be shown to be equal to the sum of the second and third terms from $ t_i^l $, respectively. Starting from the second term
\begin{subequations}
\begin{gather}
\sum_{i \in \mN^l} \left(p_i^l - \bar{p}_{-i}^l\right)^2 = \sum_{i \in \mN^l} (p_i^l)^2 + \left(\bar{p}_{-i}^l\right)^2 - 2 p_i^l \bar{p}_{-i}^l \\
= \frac{N^l}{N^l - 1} \sum_{i \in \mN^l} (p_i^l)^2 - \frac{2N^l}{(N^l - 1)^2} \sum_{\substack{k,m \in \mN^l \\ k > m}} p_k^l p_m^l \\
= \frac{N^l}{(N^l - 1)^2} \Big[(N^l - 1) \sum_{i \in \mN^l} (p_i^l)^2 - 2 \sum_{\substack{k,m \in \mN^l \\ k > m}} p_k^l p_m^l \Big] = \frac{N^l}{(N^l - 1)^2} \sum_{\substack{k,m \in \mN^l \\ k > m}} \left(p_k^l - p_m^l\right)^2
\end{gather}
\end{subequations}
Now observing
\begin{equation}
\sum_{\substack{k,m \in \mN^l \\ k > m}} \left(p_k^l - p_m^l\right)^2 = \frac{1}{N^l-2} \sum_{i \in \mN^l} \sum_{\substack{k,m \in \mN^l \\ k > m \\ k,m \ne i}} \left(p_k^l - p_m^l\right)^2
\end{equation}
we get the equality of second terms. For the third term
\begin{subequations}
\begin{gather}
\sum_{i \in \mN^l} \etas p_j^l \bar{p}_{-j}^l \Big(c_l - \sum_{k \in \mN^l} \alpha_k^l y_k\Big)^2 = \etas c_l^2 \sum_{j \in \mN^l} p_j \bar{p}_{-j}^l + \sum_{j \in \mN^l} \etas p_j \bar{p}_{-j}^l \Big(\sum_{k \in \mN^l} \phi(y_k) \Big) \\
+ {} 2 \eta \sum_{j \in \mN^l} p_j \bar{p}_{-j}^l \Big( \sum_{k > m} A_{kl} A_{ml} y_k y_m \Big) \equiv \eta \left( T_1 + T_2 + T_3 \right)
\end{gather}
\end{subequations}
where as before $ \phi(y_k) = A_{kl}^2 y_k^2 - 2 c_l A_{kl} y_k $. Here we will equate $ T_1,T_2,T_3 $ with the sum of $ f_{i,3a}^l, f_{i,3b}^l, f_{i,3c}^l $, respectively.
\begin{gather}
c_l^2 \sum_{j \in \mN^l} p_j \bar{p}_{-j}^l = \frac{2c_l^2}{N^l-1} \sum_{k>m} p_k^l p_m^l = \frac{2c_l^2}{(N^l-1)(N^l-2)} \sum_{i \in \mN^l} \sum_{\substack{k>m \\ k,m \ne i}} p_k^l p_m^l
\end{gather}
This completes $ f_{i,3a} $, now for $ f_{i,3b} $
\begin{subequations}
\begin{gather}
\sum_{j \in \mN^l} p_j \bar{p}_{-j}^l \Big(\sum_{k \in \mN^l} \phi(y_k) \Big) = \frac{2}{N^l-1} \sum_{j>m} \sum_{k \in \mN^l} p_j^l p_m^l \phi(y_k)
\\
= \frac{2}{N^l - 1} \sum_{j>m} \Big(\sum_{k \ne j,m}  p_j^l p_m^l \phi(y_k) + \sum_{k=j,m}  p_j^l p_m^l \phi(y_k)  \Big)
\\
= \frac{2}{N^l - 1} \sum_{i \in \mN^l} \sum_{\substack{j>m \\ j,m \ne i}} \Bigg( \frac{1}{N^l - 3}  \sum_{\substack{k \ne i,j,m}}  p_j^l p_m^l \phi(y_k) + \frac{1}{N^l - 2}  \sum_{k=j,m}  p_j^l p_m^l \phi(y_k) \Bigg)
\end{gather}
\end{subequations}
This completes $ f_{i,3b} $ and finally the sum for $ f_{i,3c} $ can be shown in the same way as above.

Hence after comparing the sum of all three terms from $ t_i^l $ we have indeed proved~\eqref{EQtaxred2}.
\end{proof}

Now just as before, we check individual rationality (for the modified game $ \mathfrak{G}^{\prime\prime} $). Overall utility for agent $ i $ after redistribution, at NE, is
\begin{equation} \label{EQsbbu}
\widehat{u}_i = v_i(x_i) - x_i \sum_{l \in \mL_i} A_{li} p^l + \sum_{l \in \mL_i} f_{i,1}^l + f_{i,2}^l + f_{i,3}^l
\end{equation}
where the last summation is calculated at NE.
At NE, we have equal prices, hence from~\eqref{EQfi2} it is clear that $ \sum_{l \in \mL_i} f_{i,2}^l $ term is zero. Since we are only dealing with the non-degenerate case in this section, it is clear from the definition of $ f_{i,1}^l $ (see~\eqref{EQfi1}) that $ \sum_{l \in \mL_i} f_{i,1}^l \ge 0 $. We know already, from Lemma~\ref{lemir}, that $ v_i(x_i) - x_i \sum_{l \in \mL_i} A_{li}p^l - v_i(0) > 0 $, furthermore since the functions $ v_i $ are strictly concave and $ x_i > d_i > 0 $ (by assumption (A2)) we can bound the difference in above inequality by a value bigger than zero.

Now note that $ f_{i,3}^l $ is multiplied by a positive factor $ \eta $ which is still to be chosen. Since prices $ p^l $, demand/allocation $ y_i $ and coefficients $ A_{li} $ are all absolutely bounded, we can always make the contribution of $ \sum_{l \in \mL_i} f_{i,3}^l $ small enough (without having the knowledge of equilibrium values in advance) so that the overall utility in~\eqref{EQsbbu} is bigger than $ v_i(0) $. Hence we have individual rationality here as well.

\section{Generalizations} \label{secgen}

Three quite interesting generalizations arise immediately from the set-up and analysis in this paper. The first is a case where agents have utilities based on a vector allocation rather than a scalar allocation i.e., the multiple goods scenario. The objective in this case will be written as
%\todo{Vector generalization to be emphasized more}
\begin{equation}
\sum_{i \in \mN} v_i(\underline{x}_i) \quad \text{with} \quad \underline{x}_i \in \mathbb{R}_+^{D_i}. %\{x_1,\ldots,x_N\}
\end{equation}
Note that the assumption of strict concavity can still be made. In a communications scenario, such an example can arise if the Internet agents have utility based on throughput as well as delay or packet error rate.
In the case of the local public goods problem and in particular its instance that models the interaction between wireless agents,
vector allocation for agent $i$ models the power level of all users that can interfere with $i$.
In this case the problem~\eqref{CPlpg} can be restated as
\begin{subequations}
\begin{gather}
\max_{x}~ \sum_{i \in \mN} v_i(\underline{x}_i) \tag{C$ _{\text{vlpb}} $} \\
\text{s.t.} \quad \underline{x}_i \in \mathbb{R}_+^{D_i} ~~\forall~i \in \mN \\
\text{s.t.} \quad {E_i^k}^{\top} \underline{x}_i = {E_j^k}^{\top} \underline{x}_j ~~\forall~i,j \in \mN_k,~\forall~k \in \mK
\label{eq:vec_eq}
\end{gather}
\end{subequations}
where $ \mN_k $ as before denote localities, however they needn't be disjoint now. Note that in the vector equations~\eqref{eq:vec_eq}, multiplication by matrices
$ E_i^k $ accomplishes the task of selecting some coordinates from $\underline{x}_i$.

The second generalization is with problems which can be equivalently formulated in the form of~\eqref{CPg}, but perhaps with the help of auxiliary variables. Consider the multi-rate multicast system, in which agents are divided into multicast groups, where agents within a multicast group communicate with exactly the same data but possibly at different quality (information rate) for each agent. This problem has both - private and public goods characteristics. Just like unicast, the limited capacity on links creates constraints on allocation here. However for saving bandwidth, only the highest demanded rate from each multicast group is transmitted on every link. This means that resources are directly shared within every group.
If we index agents by a double-index $ ki \in \mN $ where $ k $ represents their multicast group and $ i $ is the sub-index within the multicast group $ k $, we can state the optimization problem as
\begin{gather}
\max_{x \in \mathbb{R}_+^N}~ \sum_{ki \in \mN} v_{ki}(x_{ki}) \tag{CP$ _\text{m} $} \\
\text{s.t.} \quad \sum_{k \in \mK^l} \max_{i \in \mG_k^l}\, \{\alpha_{ki}^l x_{ki}\} \le c^l \quad \forall~l \in \mL
\end{gather}
here $ \mK^l $ represents the set of multicast groups present on link $ l $ and $ \mG_k^l $ represents the set of sub-indices from multicast group $ k $, that are present on link $ l $. Mathematically, the important property to note above is that the feasibility region is indeed a polytope (since the constraints are piecewise linear). We can easily linearize the above by using auxiliary variables $ m_k^l $ as proxies for $ \displaystyle\max_{i \in \mG_k^l}\, \{\alpha_{ki}^l x_{ki}\} $ indexed by group and link as follows
\begin{subequations}
\begin{gather}
\max_{x,m}~ \sum_{ki \in \mN} v_{ki}(x_{ki}) \tag{CP$^{\prime}_\text{m} $} \\
\text{s.t.} \quad \alpha_{ki}^l x_{ki} \le m_k^l \quad \forall~i \in \mG_k^l,~k \in \mK^l,~l \in \mL \\
\text{s.t.} \quad \sum_{k \in \mK^l} m_k^l \le c^l \quad \forall~l \in \mL.
\end{gather}
\end{subequations}
In the presence of auxiliary variables, the construction of the mechanism (esp. taxes) are slightly different in nature but follow the same philosophy as in this paper. Readers may refer to \cite{demosmulticorrection}, \cite{SiAn_multicast_arxiv} for a full implementing mechanism specifically for the multi-rate multicast problem. Other works for the multicast problem include finding decentralized algorithms to achieve social utility maximizing allocation \cite{sarkaropt}, \cite{stoenescu2007multirate} as well as determining optimal allocation via max-min fairness \cite{sarkarfair}, \cite{SaTa02}.
The incorporation of this model into the unified design methodology is a research topic the authors are currently working on.

The third possible generalization would be where the constraint set, instead of being a polytope, is taken as general convex set (satisfying assumptions from Section~\ref{subsecassump}). If one represents the centralized problem in the form a general convex optimization problem
\begin{subequations}
\begin{gather}
\max_{x,m}~ \sum_{i \in \mN} v_i(x_i) \\
\text{s.t.} \quad g_l(x) \le 0 \quad \forall~l \in \mL,
\end{gather}
\end{subequations}
then, looking at the Lagrangian
\begin{equation}
L(x,\lambda) = \sum_{i \in \mN} v_i(x_i) - \sum_{l \in \mL} \lambda_l g_l(x),
\end{equation}
one can construct the first term (payment) in the tax (see~\eqref{EQtaxg}) in a similar way as before. After this, the rest of tax terms can be analogously constructed from the KKT conditions.

The aforementioned generalizations together can lead to fully implementing mechanisms with minimal message space that can solve an even larger class of problems of interest.

We conclude with a note on the robustness properties of these mechanisms that we would like to investigate.
Recently there has been a lot of work related to robustness in mechanism design\footnote{see \cite{fudenbergrobustness} for treatment of the same and \cite{bergemann2013introduction} for a survey of the major robustness results in Mechanism design.}. In addition to investigating the learning properties of our mechanism, we are also interested in investigating robustness of our mechanism (w.r.t. information i.e. beliefs of agents) and how the two might be related.

%\newpage
%\bibliographystyle{acmsmall}
%\bibliographystyle{alpha}
%\nocite{*}
\bibliographystyle{unsrtnat}
\bibliography{abhinav}
%\input{citations.bbl}
%\bibliographystyle{alpha}
%\nocite{*}
%\bibliography{plg}

%\end{spacing}
\end{document}